\documentclass[a4paper,UKenglish]{lipics-v2021}

\nolinenumbers
\hideLIPIcs

\usepackage{amsmath}%
\usepackage{amssymb}%
\usepackage{graphicx}%
\usepackage{color}%
\usepackage{xspace}%
\usepackage{mleftright}%
\usepackage{xcolor}%
\usepackage{caption}%
\usepackage{stmaryrd}%

\usepackage[algo2e,boxed,linesnumbered,noend]{algorithm2e}
\usepackage[noend]{algpseudocode}

\usepackage{hyperref}%

\numberwithin{figure}{section}%
\numberwithin{table}{section}%
\numberwithin{equation}{section}%

    \newtheorem{problem}[theorem]{Problem}%

\newcommand{\HLinkShort}[2]{\hyperref[#2]{#1\ref*{#2}}}
\newcommand{\HLink}[2]{\hyperref[#2]{#1~\ref*{#2}}}
\newcommand{\HLinkPage}[2]{\hyperref[#2]{#1~\ref*{#2}%
      $_\text{p\pageref{#2}}$}}
\newcommand{\HLinkPageOnly}[1]{\hyperref[#1]{Page~\refpage*{#1}%
      $_\text{p\pageref{#1}}$}}

\newcommand{\HLinkSuffix}[3]{\hyperref[#2]{#1\ref*{#2}{#3}}}
\newcommand{\HLinkPageSuffix}[3]{\hyperref[#2]{#1\ref*{#2}%
      #3$_\text{p\pageref{#2}}$}}

\newcommand{\deflab}[1]{\label{def:#1}}
\newcommand{\defref}[1]{\HLink{Definition}{def:#1}}%
      
\newcommand{\figlab}[1]{\label{fig:#1}}
\newcommand{\figref}[1]{\HLink{Figure}{fig:#1}}

\newcommand{\obslab}[1]{\label{obs:#1}}
\newcommand{\obsref}[1]{\HLink{Observation}{obs:#1}}%

\providecommand{\lemlab}[1]{\label{xlemma:#1}}
\renewcommand{\lemlab}[1]{\label{xlemma:#1}}
\newcommand{\lemref}[1]{\HLink{Lemma}{xlemma:#1}}%

\newcommand{\problab}[1]{\label{problem:#1}}
\newcommand{\probref}[1]{\HLink{Problem}{problem:#1}}%

\newcommand{\corlab}[1]{{\label{cor:#1}}}
\newcommand{\corref}[1]{\HLink{Corollary}{cor:#1}}

\newcommand{\algolab}[1]{{\label{algo:#1}}}
\newcommand{\algoref}[1]{\HLink{Algorithm}{algo:#1}}

\newcommand{\thmlab}[1]{{\label{theo:#1}}}
\newcommand{\thmref}[1]{\HLink{Theorem}{theo:#1}}

\newcommand{\seclab}[1]{\label{sec:#1}}
\newcommand{\secref}[1]{\HLink{Section}{sec:#1}}

\providecommand{\eqlab}[1]{}%
\renewcommand{\eqlab}[1]{\label{equation:#1}}

\renewcommand{\Re}{\mathbb{R}}%
\newcommand{\eps}{{\varepsilon}}%

\newcommand{\CH}{{\mathcal{CH}}}%

\DefineNamedColor{named}{RedViolet} {cmyk}{0.07,0.90,0,0.34}

\newcommand{\myparagraph}[1]{\bigskip\noindent{\textbf{#1}}}

\newcommand{\distX}[2]{||#1-#2||}

\title{Fast and Exact Convex Hull Simplification}

\author{Georgiy Klimenko}{Department of Computer Science; University of Texas at Dallas; Richardson, TX 75080, USA}{gik140030@utdallas.edu}{}{}
\author{Benjamin Raichel}{Department of Computer Science; University of Texas at Dallas; Richardson, TX 75080, USA}{benjamin.raichel@utdallas.edu}{}{}

\authorrunning{G. Klimenko, and B. Raichel}
\Copyright{Georgiy Klimenko, and Benjamin Raichel}



\funding{Partially supported by NSF CAREER Award 1750780.}

\acknowledgements{The authors thank Sariel Har-Peled for helpful discussions related to the paper.}

\EventEditors{Miko{\l}aj Boja\'{n}czyk and Chandra Chekuri}
\EventNoEds{2}
\EventLongTitle{41st IARCS Annual Conference on Foundations of Software Technology and Theoretical Computer Science (FSTTCS 2021)}
\EventShortTitle{FSTTCS 2021}
\EventAcronym{FSTTCS}
\EventYear{2021}
\EventDate{December 15--17, 2021}
\EventLocation{Virtual Conference}
\EventLogo{}
\SeriesVolume{213}
\ArticleNo{8}

\begin{document}
\maketitle

\begin{abstract}
 Given a point set $P$ in the plane, we seek a subset  $Q\subseteq P$, whose convex hull gives a smaller and thus simpler representation of the convex hull of $P$. Specifically, let $cost(Q,P)$ denote the Hausdorff distance between the convex hulls $\CH(Q)$ and $\CH(P)$. Then given a value $\eps>0$ we seek the smallest subset $Q\subseteq P$ such that $cost(Q,P)\leq \eps$. We also consider the dual version, where given an integer $k$, we  seek the subset $Q\subseteq P$ which minimizes $cost(Q,P)$, such that $|Q|\leq k$. For these problems, when $P$ is in convex position, we respectively give an $O(n\log^2 n)$ time algorithm and an $O(n\log^3 n)$ time algorithm, where the latter running time holds with high probability.
 When there is no restriction on $P$, we show the problem can be reduced to APSP in an unweighted directed graph, yielding an $O(n^{2.5302})$ time algorithm when minimizing $k$ and an $O(\min\{n^{2.5302}, kn^{2.376}\})$ time algorithm  when minimizing $\eps$, using prior results for APSP. Finally, we show our near linear algorithms for convex position give 2-approximations for the general case.
\end{abstract}

\section{Introduction}

The convex hull of a set of points in the plane is one of the most well studied objects in computational geometry. As the number points on the convex hull can be linear, for example when the points are in convex position, it is natural to seek the best simplification using only $k$ input points. 
To measure the quality of the subset we use one of the most common measures, namely the Hausdorff distance.
Specifically, given a set $P$ of $n$ points in the plane, here we seek the subset of $Q\subseteq P$ of $k$ points which minimizes the Hausdorff distance between $\CH(Q)$ and $\CH(P)$, where $\CH(X)$ denotes the convex hull of $X$.
This is equivalent to finding the subset $Q\subseteq P$ of $k$ points which minimizes $\eps=\max_{p\in P} ||p-\CH(Q)||$. We refer to this as the \emph{min-$\eps$} problem.
We also consider the dual \emph{min-$k$} problem, where given a distance $\eps\geq 0$, we seek the minimum cardinality subset $Q\subseteq P$ such that $\max_{p\in P} ||p-\CH(Q)||\leq \eps$. 
We emphasize that our goal is to find the optimal subset $Q$ exactly. As discussed below, this is a far more demanding problem than allowing approximation in terms of $k$ or $\eps$.

A number of related problems have been considered before, though they all differ in key ways. The three main differences concern the error measure of $Q$, whether $Q$ is restricted to be a subset from $P$, and whether a starting point is given. Varying any one of these aspects can significantly change the hardness of the problem.

\myparagraph{Coresets.}
In this paper, we require our chosen points to be a subset of $P$, which from a representation perspective is desirable as the chosen representatives are actual input data points. 
Such subset problems have thus been extensively studied, and are referred to as coresets (see \cite{handbook3}).
Given a point set $P$, a coreset is subset of $P$ which approximately preserves some geometric property of the input. Thus here we seek a coreset for the Hausdorff distance. 

Among coreset problems, $\eps$-kernels for directional width are one of the most well studied. Define the directional width for a direction $u$ as 
$w(u,P) = \max_{p\in P} \langle u,p\rangle - \min_{p\in P} \langle u,p\rangle$. Then $Q\subseteq P$ is an $\eps$-kernel if for all $u$, $(1-\eps)w(u,P)\leq w(u,Q)$. 
It is known that for any point set $P\subset \Re^d$ there is an $\eps$-kernel of size $O(1/\eps^{(d-1)/2})$ \cite{ahv-aemp-04}. For worst case point sets $\Omega(1/\eps^{(d-1)/2})$ size is necessary, however, for certain inputs, significantly smaller coresets may be possible. (As an extreme example, if the points lie on a line, then the $k=2$ extreme points achieves $\eps=0$ error.)
Thus \cite{bhr-sagps-19} considered computing coresets whose size is measured relative to the optimum for a given input point set. Specifically, if there exists an $\eps$-coreset for Hausdorff distance with $k$ points, then in polynomial time they give an $\eps$-coreset with $O(d k\log k)$ size, or alternatively an $(8\eps^{1/3}+\eps)$-coreset with $O(k/\eps^{2/3})$ size. 
Note that the standard strategy to compute $\eps$-kernels applies a linear transformation to make the point set fat, and then roughly speaking approximates the Hausdorff problem. 
Thus $\eps$-coresets for Hausdorff distance yield $O(\eps)$-kernels (where the constant relates to the John ellipsoid theorem). 
However, $\eps$-kernels do not directly give such coresets for Hausdorff distance, as it depends on the fatness of the point set, i.e.\ Hausdorff is arguably the harder problem.

Most prior work on coresets gave approximate solutions. However, our focus is on exact solutions. Along these lines, a very recent PODS paper \cite{wmlt-mcmrmd-21} considered what they called the \emph{minimum $\eps$-corset} problem, where the goal is to exactly find the minimum sized $\eps$-coreset for a new error measure they introduced. 
Specifically, $Q\subseteq P$ is an $\eps$-coreset for maxima representation if for all directions $u$, $(1-\eps)\omega(u,P)\leq \omega(u,Q)$, where $\omega(u,X) = \max_{x\in X} \langle u,x\rangle $.
While related to our Hausdorff measure, again like directional width, it differs in subtle ways. For example, observe their measure is not translation invariant. Moreover, they assume the input is $\alpha$-fat for some constant $\alpha$, while we do not. 
For their measure they give a cubic time algorithm in the plane, whereas our focus is on significantly subcubic time algorithms.

In the current paper, we select $Q$ so as to minimize the maximum distance of a point in $P$ to $\CH(Q)$. \cite{krv-schc-20} instead considered the problem of selecting $Q$ so as to minimize the sum of the distances of points in $P$ to $\CH(Q)$. They provided near cubic (or higher) running time algorithms for certain generalized versions of this summed coreset variant.

\myparagraph{Other related problems.}
If one relaxes the problem to no longer require $Q$ to be a subset of $P$, then related problems have been studied before. 
Given two convex polygons $X$ and $Y$, where $X$ lies inside $Y$, \cite{abosy-fmcnp-89} provided a near linear time algorithm for the problem of finding the convex polygon $Z$ with the fewest number of vertices such that $X\subseteq Z\subseteq Y$. 
The problem of finding the best approximation under Hausdorff distance has also been considered before. Specifically, if $Q$ can be any subset from $\CH(P)$ (i.e.\ it is not a coreset), then \cite{ar-hacp-05} gave a near linear time algorithm for approximating the convex hull under Hausdorff distance, but under the key assumption that they are given a starting vertex which must be in $Q$. We emphasize that assuming a starting point is given makes a significant difference, and intuitively relates to the difference in hardness between single source shortest paths and all pairs shortest paths. 

A number of papers have considered simplifying polygonal chains. Computing the best global Hausdorff simplification is NP-hard \cite{klw-opshf-20,kklmw-gcs-19}. Most prior work instead considered local simplification, where points from the original chain are assigned to the edge of the simplification whose end points they lie between. 
In general such algorithms take at least quadratic time, with subquadratic algorithms known for certain special or approximate cases. For example, \cite{av-eaapc-00} gave an $O(n^{4/3+\delta})$ time algorithm, for any $\delta>0$, under the $L_1$ metric. Our problem relates to these works in that we must approximate the chain representing the convex hull. On the one hand, convexity gives us additional structure. However, unlike polygonal chain simplification, we do not have a well defined starting point (i.e.\ the convex hull is a closed chain), which as remarked above makes a significant difference in hardness. 

Our problem also relates to polygon approximation, for which prior work often instead considered approximation in relation to area. For example, given a convex polygon $P$, \cite{hklmu-matcpr-20} gave a near linear time algorithm for finding the three vertices of $P$ whose triangle has the maximum area. To illustrate one the many ways that area approximations differ, observe that the area of the triangle of the three given points of $P$ can be determined in constant time, whereas the computing the furthest point from $P$ to the triangle takes linear time.

\myparagraph{Our results.}
We give fast and exact algorithms for both the min-$k$ and min-$\eps$ problems for summarizing the convex hull in the plane. While a number of related problems have been considered before as discussed above, to the best of our knowledge we are the first to consider exact algorithms for this specific version of the problem.

Our main results show that when the input set $P$ is in convex position then the min-$k$ problem can be solved exactly in $O(n\log^2 n)$ deterministic time, and the min-$\eps$ problem can be solved exactly in $O(n \log^3 n)$ time with high probability. 
Note that this version of the problem is equivalent to allowing the points in $P$ to be in arbitrary position, but requiring that the chosen subset $Q$ consist of vertices of the convex hull. (Which follows as the furthest point to $\CH(Q)$ is always a vertex of $\CH(P)$.) Thus this restriction is quite natural, as we are then using vertices of the convex hull to approximate the convex hull, i.e.\ furthering the coreset motivation.

For the general case when $P$ is arbitrary and $Q$ is any subset of $P$, we show that in near quadratic time these problems can be reduced to computing all pairs shortest paths in an unweighted directed graph.  This yields an $O(n^{2.5302})$ time algorithm for the min-$k$ problem and an $O(\min\{n^{2.5302}, kn^{2.376}\})$ time algorithm for the min-$\eps$ problem, by utilizing previous results for APSP in unweighted directed graphs.
Moreover, while exact algorithms are our focus, we show that 
our near linear time algorithms for points in convex position immediately yield 2-approximations for the general case with the same near linear running times. Also, appropriately using single source shortest paths rather than APSP in our graph based algorithms, gives $O(n^2\log n)$ time solutions which use at most one additional point.

\section{Preliminaries}

Given a point set $X$ in $\Re^2$, let $\CH(X)$ denote its convex hull. 
For two points $x,y \in \Re^2$, let $\overline{xy}$ denote their line segment, that is $\overline{xy}=\CH(\{x,y\})$.
Throughout, given points $x,y\in \Re^2$, $||x-y||$ denotes their Euclidean distance.  
Given two compact sets $X,Y\subset \Re^2$, $\distX{X}{Y} = \min_{x\in X, y\in Y} ||x-y||$ denotes their distance. For a single point $x$ we write $\distX{x}{Y} = \distX{\{x\}}{Y}$.

For any two finite point set $Q,P\subset \Re^2$ we define 
\[
cost(Q,P) = \max_{p \in P} ||p-\CH(Q)||
\]
Note that for $Q\subseteq P$, we have that $\CH(Q)\subseteq \CH(P)$, and moreover the furthest point in $\CH(P)$ from $\CH(Q)$ is always a point in $P$. Thus the $cost(Q,P)$ is equivalent to the Hausdorff distance between $\CH(Q)$ and $\CH(P)$. 

In this paper we consider the following two related problems, where for simplicity, we assume that $P$ is in general position.

\begin{problem}[min-$k$]\problab{mink}
Given a set $P \subset \Re^2$ of $n$ points, and a value $\eps>0$, find the smallest integer $k$ such that there exists a subset $Q\subseteq P$ where $|Q|\leq k$ and $cost(Q,P)\leq \eps$.
\end{problem}

\begin{problem}[min-$\eps$]\problab{mineps}
Given a set $P \subset \Re^2$ of $n$ points, and an integer $k$, find the smallest value $\eps$ such that there exists subset $Q\subseteq P$ where $|Q|\leq k$ and $cost(Q,P)\leq \eps$.
\end{problem}
For simplicity the above problems are phrased in terms of finding the value of either $k$ or $\eps$, though we remark that our algorithms for these problems also immediately imply the set $Q$ realizing the value can be determined in the same time. Thus in the following when we refer to a solution to these problems, we interchangeably mean either the value or the set realizing the value.

In the following section we restrict the point set $P$ to lie in convex position, thus for simplicity we define the following convex versions of the above problems.

\begin{problem}[cx-min-$k$]\problab{cxmink}
Given a set $P \subset \Re^2$ of $n$ points in convex position, and a value $\eps>0$, find the smallest integer $k$ such that there exists a subset $Q\subseteq P$ where $|Q|\leq k$ and $cost(Q,P)\leq \eps$.
\end{problem}

\begin{problem}[cx-min-$\eps$]\problab{cxmineps}
Given a set $P \subset \Re^2$ of $n$ points in convex position, and an integer $k$, find the smallest value $\eps$ such that there exists subset $Q\subseteq P$ where $|Q|\leq k$ and $cost(Q,P)\leq \eps$.
\end{problem}

\section{Convex Position}\seclab{convexpos}
In this section we give near linear time algorithms for the case when $P$ is in convex position, that is for \probref{cxmink} and \probref{cxmineps}. First, we need several structural lemmas and definitions. 

\subsection{Structural Properties and Definitions}

\begin{figure}[t]
\centering
\includegraphics[scale=.55]{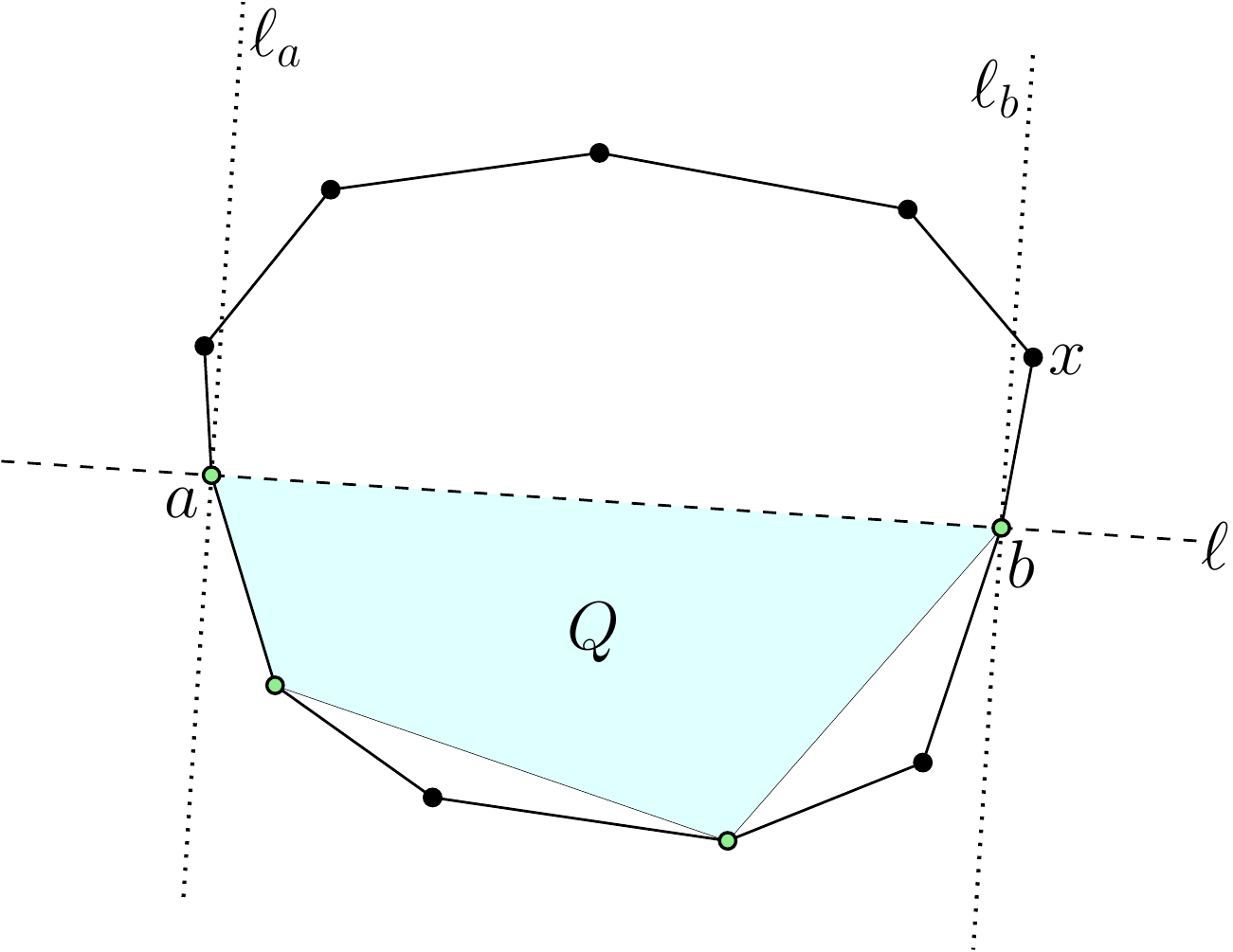}
\caption{An example of the defined objects from \lemref{helper}.}
\figlab{segments}
\end{figure}

\begin{lemma}\lemlab{helper}
Let $P$ be a set of $n$ points in convex position. Consider any subset $Q \subset P$, and let $a,b$ be consecutive in the clockwise ordering of $Q$. 
Then for any point $x\in P$ which falls between $a$ and $b$ in the clockwise ordering of $P$, we have $||x-\CH(Q)|| = ||x-\overline{ab}||$.
\end{lemma}
\begin{proof}
Let $x$ be any point between $a$ and $b$ in the clockwise ordering of $P$, and let $l$ denote the line through $a$ and $b$.
Since $a$ and $b$ are consecutive in the clockwise order of $Q$, $\CH(Q)$ lies entirely in the closed halfspace defined by $\ell$ and on the opposite side of $\ell$ as $x$. So if $z$ denotes the closest point to $x$ in $\CH(Q)$, then the segment $\overline{xz}$ must intersect $\ell$. 

Consider the lines $l_a$ and $l_b$ which are perpendicular to $l$ and go through $a$ and $b$ respectively. If $x$ lies between $l_a$ and $l_b$, then its projection onto $\ell$ lies on the segment $\overline{ab}$, and hence this is in fact its projection onto $\CH(Q)$, and the claim holds. Otherwise, suppose that $x$ and $a$ are in opposite halfplanes defined by the line $l_b$, see \figref{segments}. (A similar argument will hold when $x$ and $b$ are in opposite halfplanes defined by the line $l_a$.)
Observe, that the closest point in $\ell_b\cap \CH(Q)$ to $x$ is the point $b$, since $x$ is in the opposite halfspace defined by $\ell$ as $\CH(Q)$, and $\ell_b$ is orthogonal to $\ell$. Thus if the shortest path to $z$ intersects $\ell_b$, then it would imply $z=b$, and so again the claim holds. So suppose $z$ and $x$ are on the same side of $\ell_b$. 
Since $\overline{xz}$ intersects $\ell$, $z$ must lie on the opposite side of $\ell$ as $x$. Since $z\in \CH(Q)$, this implies there is a point $y\in Q$ which like $z$ is on the same side of $\ell_b$ as $x$ but on the opposite side of $\ell$ as $x$ (since there is no point of $Q$ on the same side of $\ell$ as $x$).  Thus similarly, the segment $\overline{xy}$ intersects $\ell$, and let $y'$ denote this intersection point. Since $x$ and $y$ are on the same side of $\ell_b$, which is opposite the side of $a$, this implies $b$ lies on the segment $\overline{ay'}$.
As $y'$ lies on the segment $\overline{xy}$, this in turn implies that $b$ lies in the triangle $\Delta(ayx)$.  
This is a contradiction, since $a,y,x,b\in P$, and so $b$ lying in $\Delta(azx)$ implies $P$ is not in convex position. 
\end{proof} 

Assume that the points in $P=\{p_1,\ldots, p_n\}$ are indexed in clockwise order. We now wish to prove a lemma about the optimal cost solution when restricted to points between some index pair $i,j$.  
As we wish our definition to work regardless of whether $i\leq j$ or $j\leq i$, we define the following notation. For a triple of indices $i,x,j$, we write $i\preceq x \preceq j$ to denote that $p_x$ falls between $p_i$ and $p_j$ in the clockwise ordering. More precisely, if $i\leq j$ then this means $i\leq x\leq j$, and if $j\leq i$ then this means that $j\leq x \leq n$ or $1\leq x\leq i$.

\begin{definition}\deflab{costg}
For any integer $0\leq k\leq n-2$ we define 
\[
cost_k(i,j) = \min_{i\preceq l_1\preceq \ldots \preceq l_k \preceq j} ~~\max_{i \preceq v \preceq j} ||p_v-\CH(p_i, p_{l_1}, \ldots, p_{l_k}, p_j)||.
\]
\end{definition}
That is, $cost_k(i,j)$ is the minimum cost solution when restricted to including $p_i$, $p_j$, and $k$ other vertices in clockwise order between $p_i$ and $p_j$, and where we only evaluate the cost with respect to points in clockwise order between $p_i$ and $p_j$.  

According to the above definition, we have that 
$cost_0(i,j)=\max_{i \preceq v \preceq j} ||p_v-\CH(p_i, p_j)|| = \max_{i \preceq v \preceq j} ||p_v-\overline{p_ip_j}||$. Observe that the following is implied by \lemref{helper}.

\begin{corollary}\corlab{decompose}
Let $Q=\{p_{l_1},\ldots,p_{l_k}\}\subseteq P$ be indexed in clockwise order, and let $l_{k+1}=l_1$. Then,
\[
cost(Q,P) = \max_{p \in P} ||p-\CH(Q)||
=\max_{1\leq i\leq k}~~ \max_{l_i\preceq j\preceq l_{i+1}} ||p_j-\overline{p_{l_i}p_{l_{i+1}}}|| =\max_{1\leq i\leq k} cost_0(l_i,l_{i+1}).
\]
\end{corollary}
For more general values of $k$, the following lemma will be used to argue we can use a greedy algorithm.

\begin{lemma}\lemlab{ranges}
For any indices $i'\preceq i\preceq j \preceq j'$, it holds that $cost_k(i,j) \leq cost_k(i', j')$.
\end{lemma}
\begin{proof}
Let $p_{i'}, p_{l_1}, \ldots p_{l_k}, p_{j'}$ be the clockwise chain of vertices that realizes $cost_k(i', j')$. 
That is, $cost_k(i', j')=max_{i' \preceq v \preceq j'} ||p_v-\CH(p_{i'}, p_{l_1}, \ldots p_{l_k}, p_{j'})||$. 
Observe that if we add points to this chain then we can only decrease the cost. Specifically, we consider adding the points $p_i$ and $p_j$. So let $p_{l_x},\ldots, p_{l_y}$ be the subchain of $p_{l_1}, \ldots p_{l_k}$ consisting of all $i\preceq l_i\preceq j$. Then we have, 
\begin{align*}
cost_k(i', j')
&=max_{i' \preceq v \preceq j'} ||p_v-\CH(p_{i'}, p_{l_1}, \ldots, p_{l_k}, p_{j'})||\\
&\geq max_{i' \preceq v \preceq j'} ||p_v-\CH(p_{i'}, p_{l_1},\ldots, p_i, p_x, \ldots, p_y, p_j, \ldots, p_{l_k}, p_{j'})||\\
&\geq max_{i \preceq v \preceq j} ||p_v-\CH(p_{i'}, p_{l_1},\ldots, p_i, p_x, \ldots, p_y, p_j, \ldots, p_{l_k}, p_{j'})||\\
&\geq max_{i \preceq v \preceq j} ||p_v-\CH(p_i, p_x, \ldots, p_y, p_j)||
\geq cost_k(i,j).
\end{align*}
The second to last inequality holds by \lemref{helper}.
The last inequality holds as the chain $p_x, \ldots, p_y$ has at most $k$ points (since it was a subchain of $p_{l_1},\ldots, p_{l_k}$) and $cost_k(i,j)$ is defined by the minimum cost such chain between $i$ and $j$. 
\end{proof}

We now define the notions of friends and greedy sequences, which we use in the next section to design our greedy algorithm.

\begin{definition}\deflab{friends}
 For an index $i$ and value $\eps\geq 0$, define the $\eps$-friend of $i$, denoted $f_\eps(i)$, as the index $j$ of the vertex furthest from $p_i$ in the clockwise ordering of $P$, such that $cost_0(i,j)\leq \eps$.
\end{definition}
Note that $f_\eps(i)$ is always well defined. In particular, $cost_0(i,i+1)=0$ for any $i$. 
Moreover, if the ball of radius $\eps$ centered at $p_i$ contains $P$ then $f_\eps(i)=i$, and the point $p_i$ by itself is an optimal solution to \probref{cxmink}. 
Note that we can easily determine if such a point exists in $O(n\log n)$ time by computing the farthest Voronoi diagram of $P$,%
\footnote{
The farthest Voronoi diagram of $P$ partitions the plane into regions sharing the same farthest point in $P$. It allows one to find the farthest point in $P$ from a query in logarithmic time.
See for example \cite{grt-hdcg-04}.
}
and then querying all points in $P$. For simplicity we will assume $f_\eps(i)\neq i$, which can thus be assured by such a preprocessing step.

\begin{definition}\deflab{greedyseq}
Let $Q=\{p_{l_1},p_{l_2},\ldots,p_{l_k}\}$ be any subset of $P$, which we assume has been indexed such that $l_1<l_2<\ldots<l_k$. We call $Q$ a \emph{greedy sequence} if for all $1\leq i< k$, we have $f_\eps(l_i)=l_{i+1}$, and $f_\eps(l_k)<l_k$. We call a greedy sequence \emph{valid} if $f_\eps(l_k)\geq l_1$. 
\end{definition}

Note that in the above definition, the condition that $f_\eps(l_k)<l_k$ ensures that the $\eps$-friend of $p_{l_k}$ goes past the vertex $p_n$, i.e.\ this ensure that the sequence is a maximal sequence without wrapping around.
Note also there always exists a valid greedy sequence. Specifically, we trivially have that for any greedy sequence $f_\eps(l_k)\geq 1$. Thus the greedy sequence starting at $p_1$ is valid as in that case $l_1 = 1$.

\begin{observation}\obslab{greedyvalid}
 Let $Q=\{p_{l_1},p_{l_2},\ldots,p_{l_k}\}$ be a valid greedy sequence. Then since $Q$ is a greedy sequence $cost_0(l_i,l_{i+1})\leq \eps$ for all $1\leq i< k$. Furthermore,  $cost_0(l_k,l_{1})\leq cost_0(l_k,f_\eps(l_k))\leq \eps$ by \lemref{ranges} and since $Q$ is valid. Thus by \corref{decompose}, $cost(Q,P)\leq \eps$.
\end{observation}

\begin{lemma}\lemlab{greedyopt}
 Let $P,\eps$ be an instance \probref{cxmink}. Any valid greedy sequence of minimum possible cardinality is an optimal solution to the given instance.
\end{lemma}

\begin{proof}
Let $Q=\{p_{l_1},p_{l_2},\ldots,p_{l_{k}}\}$ be an optimal solution to \probref{cxmink}, indexed such that $1 \leq l_1 < l_2 < \ldots < l_{k}$. Thus $cost(Q,P)\leq \eps$ and so by \corref{decompose}, $\max_{1\leq i\leq k} cost_0(l_i,l_{i+1}) \allowbreak \leq \eps$, where $l_{k+1}=l_1$. Thus if $Q$ is a greedy sequence then it is a valid greedy sequence, and the claim holds.
So suppose $Q$ is not a greedy sequence. Now we show that $Q$ can be converted to a valid greedy sequence with the same cardinality.

Let $j>1$ be the first index such that $l_j \neq f_\eps(l_{j-1})$. Let $w_j= f_\eps(l_{j-1})$ and let $\{w_{j+1}, w_{j+2}, \ldots, w_{k}\}$ be the indices which realize $cost_{k-j}(w_j, l_1)$ according to \defref{costg}. Then we modify $Q$ by replacing the suffix $\{p_{l_j},p_{l_{j+1}}, \ldots, p_{l_{k}}\}$ with $\{p_{w_j},p_{w_{j+1}}, \ldots, p_{w_{k}}\}$.
Notice that the cost of $Q$ after this modification is still $\leq \eps$ because $cost_0(l_{j-1}, w_j) \leq \eps$ as $w_j= f_\eps(l_{j-1})$, and by \lemref{ranges} we have $cost_{k-j}(w_j,l_1) \leq cost_{k-j}(l_j,l_1)$.
Now repeat this procedure until $h = f_\eps(l_{j-1})$ goes beyond index $n$. Let the resulting new optimal solution be denoted $Q'$. If $h \geq l_1$, then $Q'$ is a valid greedy sequence by our construction, and we are done. So if the sequence failed to be a valid greedy sequence, then $1 \leq h < l_1$. Thus we can repeat the whole procedure, relabeling vertices of $Q'$ such that $l_1=h$. This means that each time we repeat this procedure we either produce a valid greedy sequence or we decrease $l_1$. 
At some point $l_1=1$, at which time the procedure must produce a valid greedy sequence as in this case $h\geq 1=l_1$.

The above argues that some valid greedy sequence of minimum cardinality is optimal. Note this implies all valid greedy sequences of minimum cardinality are optimal, since they all have the same size, and by 
\obsref{greedyvalid} their cost is $\leq \eps$.
\end{proof}

\subsection{The min-$k$ Algorithm}
In this section we give an efficient algorithm for \probref{cxmink}.
The idea is to use the $f_\eps(i)$ values to define a graph. Specifically, the \emph{friend graph} $G_f$ is the directed graph with vertex set $P$ where there is an edge from $p_i$ to $p_j$ if and only if $f_\eps(i)=j$ and $i<j$. Thus every vertex in $G_f$ has outdegree at most $1$. Moreover, $G_f$ is acyclic since we only created edges from lower index vertices to higher index ones. These two properties together imply that $G_f$ is a forest, where each sink vertex defines the root of a tree. Thus every vertex in $G_f$ has a well defined depth, where sink vertices have depth one.

Let $Q=\{p_{l_1},p_{l_2},\ldots,p_{l_k}\}$ be a greedy sequence, as defined in \defref{greedyseq}. Then observe that for all $1\leq i<k$, $p_{l_i}p_{l_{i+1}}$ is an edge of $G_f$, and hence $Q$ corresponds to a path in $G_f$. Moreover, the condition that $f_\eps(l_k)<l_k$ in \defref{greedyseq} implies that $p_{l_k}$ is a sink vertex in $G_f$, and hence $Q$ corresponds to a path in $G_f$ from the vertex $p_{l_1}$ to the root of its corresponding tree. Conversely, for the same reasons if we are given a path $p_{l_1},p_{l_2},\ldots,p_{l_k}$ in $G_f$ where $p_{l_k}$ is a sink, then this path is a greedy sequence. That is, the set of paths ending in sinks in $G_f$ and the set of greedy sequences are in one-to-one correspondence. 

Thus given all the $f_\eps(i)$ values have been precomputed, this suggests a simple linear time algorithm to compute a valid greedy sequence $Q$ with the fewest number of points, which by \lemref{greedyopt} is an optimal solution to the given instance of \probref{cxmink}.
Specifically, find all pairs $(p_i,p_r)$ where $p_i\in P$ and $p_r$ is the root of the tree in $G_f$ which contains $p_i$. 
By the above discussion, each such pair $(p_i,p_r)$ corresponds to a greedy sequence, and all greedy sequences are represented by some pair. We now restrict to pairs that are valid according to \defref{greedyseq}, that is pairs where $f_\eps(r)\geq i$. For each such pair, the length of the corresponding sequence is simply the depth of $p_i$ in the tree rooted at $p_r$. 
Thus we return as our solution the depth of $p_i$ from the valid pair $(p_i,p_r)$ where $p_i$ has minimum depth. 

All the $(p_i,p_r)$ pairs and the depths can be determined in $O(n)$ time by topologically sorting since $G_f$ is a forest. Determining the valid pairs, and the minimum depth valid pair can then be done with a simple linear scan. We thus have the following.

\begin{lemma}\lemlab{lineartime}
Assume that $f_\eps(i)$ for all $1\leq i\leq n$ has been precomputed. Then \probref{cxmink} can be solved in $O(n)$ time.
\end{lemma}

The question now then is how quickly can we compute all of the $f_\eps(i)$ values. To that end, we first argue that with some precomputation the $cost_0(i,j)$ values can be queried efficiently. To do so, we make use a result from \cite{dmsw-fpqgcc-06} which builds a datastructure for a geometric query they call Farthest Vertex in a Halfplane, which we rephrase below using our notation. 

\begin{lemma}[\cite{dmsw-fpqgcc-06}]\lemlab{cited}
Let $P\subset \Re^2$ be a point set in convex position. $P$ can be preprocessed in $O(n\log n)$ time such that given a query $(q,l_q)$, where $q$ is a point and $l_q$ is a directed line through $q$, in $O(\log^2 n)$ time one can return the farthest point from $q$ among the points in $P$ to the left of $l_q$.
\end{lemma}

\begin{lemma}\lemlab{precompConv}
Let $P=\{p_1,\ldots, p_n\} \subset \mathbb{R}^2$ be a point set in convex position, labeled in clockwise order. With $O(n \log n)$ precomputation time, for any query index pair $(i,j)$, $cost_0(i,j)$ can be computed in $O(\log^2 n)$ time.
\end{lemma}
\begin{proof}
 Let $\ell=\ell(p_i,p_j)$ be the line through $p_i$ and $p_j$, which we view as being oriented in the direction from $p_i$ towards $p_j$.
 Also, let $r_i$ and $r_j$ denote the rays originating at $p_i$ and $p_j$ respectively, pointing in the direction orthogonal to $\ell$ and on the left side side of $\ell$.
 Finally, let $P_{i,j} = \{p_k\in P \mid i\prec k\prec j\}$, and thus $cost_0(i,j) = \max_{x\in {P_{i,j}}} ||x-\overline{p_i,p_j}||$. 
 
 Observe that the projection of any point $x\in P_{i,j}$ onto $\ell$ either lies on the portion of $\ell$ before $p_i$, on the line segment $\overline{p_ip_j}$, or on the portion of $\ell$ after $p_j$. 
 Thus we have a natural partition of $P_{i,j}$ into three sets, the subset in the right angle cone $C_i$ bounded by $\ell$ and $r_i$, those in the slab $Slab(i,j)$ bounded by $\ell$, $r_i$, and $r_j$, and those in the right angle cone $C_j$ bounded by $\ell$ and $r_j$. Observe that for any point $x$ in $C_i$ or $C_j$, its closest point on $\overline{p_ip_j}$ is $p_i$ or $p_j$, respectively, and moreover $||x-\ell||\leq ||x-\overline{p_ip_j}||$. Thus we have that, 
  \begin{align*}
  cost_0(i,j) &= \max_{x\in {P_{i,j}}} ||x-\overline{p_i,p_j}||\\
 &= \max\{\max_{x\in C_i\cap P_{i,j}} ||x-p_i||, \max_{x\in C_j\cap P_{i,j}} ||x-p_j||, \max_{x\in Slab(i,j)\cap P_{i,j}} ||x-\overline{p_ip_j}||\}\\
 &= \max\{\max_{x\in C_i\cap P_{i,j}} ||x-p_i||, \max_{x\in C_j\cap P_{i,j}} ||x-p_j||, \max_{x\in P_{i,j}} ||x-\ell||\}.
 \end{align*} 
Therefore, it suffices to describe how to compute each of the three terms in the stated time. Computing $\max_{x\in P_{i,j}} ||x-\ell||$ is straightforward as the points in $P_{i,j}$ are in convex position and in particular if we consider them in their clockwise order, then their distance to $\ell$ is a concave function. 
So assume that $P$ is given in an array sorted in clockwise order. (If not, we can compute such an array with $O(n\log n)$ preprocessing time by computing the convex hull.) 
Then given a query pair $(i,j)$, in $O(\log n)$ time we can binary search over $P_{i,j}$ to find $\max_{x\in P_{i,j}} ||x-\ell||$, since $P_{i,j}$ is a subarray of $P$. (Note if $j<i$ then technically $P_{i,j}$  is two subarrays.)

Now consider the subset in the right angle cone $C_i$ (a similar argument will hold for $C_j$). Let $C_i'$ be the cone $C_i$ but reflected over the line $\ell$. Suppose that both $C_i$ and $C_i'$ contained  points from $P$, call them $p$ and $p'$, respectively. Then observe that the triangle $\Delta(p,p',p_j)$ would contain the point $p_i$, which is a contradiction as $P$ was in convex position. Thus either $C_i\cap P =\emptyset$ or $C_i'\cap P=\emptyset$. 
So let $L$ be the line orthogonal to $\ell$, passing through $p_i$, and oriented so that $C_i$ and $C_i'$ lie to the left
(i.e.\ $L$ is the line supporting the ray $r_i$ from above).
By \lemref{cited}, we can preprocess $P$ in $O(n\log n)$ time, such that in $O(\log^2 n)$ time we can compute the point in $P$ furthest from $p_i$ and to the left of $L$. If the returned point lies in $C_i'$ then we know $C_i\cap P=\emptyset$ and so $\max_{x\in C_i\cap P_{i,j}} ||x-p_i||=0$. If the returned point lies in $C_i$ then it realizes $\max_{x\in C_i\cap P_{i,j}} ||x-p_i||$.
\end{proof}

\begin{theorem}\thmlab{minkfinal}
 \probref{cxmink} can be solved in $O(n\log^2 n)$ time.
\end{theorem}

\begin{proof}
By \lemref{lineartime}, given the $f_\eps$ values have been computed, \probref{cxmink} can be solved in $O(n)$ time. Thus to prove the theorem it suffices to compute $f_\eps(i)$ for all $i$ in $O(n\log^2 n)$ time.
Recall that $f_\eps(i)$ is the index $z$ of the vertex furthest from $p_i$ in the clockwise ordering of $P$, such that $cost_0(i,z)\leq \eps$. First observe that as we increase $i$, $f_\eps(i)$ moves clockwise. More precisely, by \lemref{ranges},
$\eps \geq cost_0(i,f_\eps(i)) \geq cost_0(i+1,f_\eps(i)) \geq cost_0(i+1,j)$, for any $i+1 \leq j \leq f_\eps(i)$, and thus $i\preceq f_{\eps}(i) \preceq f_\eps(i+1)$. Moreover, again by \lemref{ranges}, the indices $j$ such that $cost_0(i,j)\leq \eps$ are consecutive in the clockwise ordering of $P$.

This suggests a simple strategy to compute the $f_\eps(i)$ values. Namely, to find $f_\eps(1)$, we compute all values  $cost(1,j)$, starting with $j=3$ and increasing $j$ until we find a value $j'$ such that $cost_0(1,j')>\eps$. This implies $f_\eps(1)=j'-1$, since as mentioned above the values such that $cost_0(1,j)\leq \eps$ are consecutive.
More generally, to compute $f_\eps(i+1)$, we compute all values  $cost_0(i+1,j)$, starting with $j=f_\eps(i)+1$ and increasing $j$ until we find a value $j'$ such that $cost_0(i+1,j')>\eps$, which again by the above implies $f_\eps(i+1)=j'-1$.

The total time is clearly bounded by the time it takes to compute all the queried $cost_0$ values. Observe that when the algorithm queries a value $cost_0(i,j)$ then the previous $cost_0$ query was either to $cost_0(i-1,j)$ or $cost_0(i,j-1)$, implying that in total we compute $O(n)$ $cost_0$ values. By \lemref{precompConv}, with $O(n\log n)$ precomputation, any $cost_0$ value can be computed in $O(\log^2 n)$ time. Thus the total time is $O(n\log^2 n)$.
\end{proof}

\subsection{The min-$\eps$ Algorithm}

In this section we design an efficient algorithm for \probref{cxmineps}, where $k$ is given and our goal is to minimize $\eps$. To do so, we will use our algorithm from the previous section for \probref{cxmink}, where $\eps$ was fixed and we were minimizing $k$. Specifically, throughout this section, given an instance $P,k$ of \probref{cxmineps}, we use $Decider(\eps)$ to denote the procedure which runs the algorithm of \thmref{minkfinal} on the instance $P,\eps$ of \probref{cxmink} and returns True if the solution found uses $\leq k$ points, and returns False otherwise. 

Let $\mathcal{E}=\{cost_0(i,j)\mid 1\leq i,j\leq n\}$. We call $\mathcal{E}$ the set of \emph{critical values}, where observe that by \corref{decompose}, the optimal solution to the given instance of \probref{cxmineps} is a critical value in the set $\mathcal{E}$. Thus a natural approach would be to explicitly compute, sort, and then binary search over $\mathcal{E}$ using $Decider(\eps)$. However, such an approach would require at least quadratic time  as $|\mathcal{E}|=\Theta(n^2)$.  We now argue that by using random sampling we can achieve near linear running time with high probability. Similar sampling strategies have been used before, and in particular we employ a strategy which was used in \cite{hr-fdre-14} for computing the Frechet distance.
We first observe that one can efficiently sample values from $\mathcal{E}$.

\begin{lemma}\lemlab{sample}
 With $O(n\log n)$ precomputation time, one can sample a value uniformly at random from $\mathcal{E}$ in $O(\log^2 n)$ time.
\end{lemma}
\begin{proof}
To sample a pair from $1\leq i,j\leq n$ uniformly at random, we first sample an integer uniformly at random from $[1,n]$ for $i$, and then sample an integer uniformly at random from $[1,n-1]$ for $j$ (where $j$ is indexed from the set with $i$ removed).
This takes $O(1)$ time given the standard assumption that sampling a random integer in a given range takes $O(1)$ time. (Even if it took $O(\log n)$ time it would not affect the overall time.)
Now to sample a value uniformly at random from $\mathcal{E}$ we just need to compute $cost_0(i,j)$. From \lemref{precompConv} this can be done in $O(\log^2 n)$ time with $O(n \log n)$ precomputation time.
\end{proof}

Before presenting our algorithm, we require the following subroutine.

\begin{lemma}\lemlab{extract}
 Given an interval $[\alpha,\beta]$, then the set $X=[\alpha,\beta]\cap \mathcal{E}$ can be computed in $O((n\log n+|X|)\log^2 n)$ time. Let $Extract(\alpha,\beta)$ denote this procedure.
\end{lemma}
\begin{proof}
 Fix an index $i$.
 By \lemref{ranges} we know that $cost_0(i,j)$ increases monotonically as we move $p_j$ clockwise. Thus $S_i=\{j\mid cost_0(i,j)\in [\alpha,\beta]\}$ is a contiguous set of indices, and moreover, we can binary search  for the smallest index in this set (i.e.\ the first index $j$ in clockwise order from $i$ such that $cost_0(i,j)\geq \alpha$). After finding this smallest such index, to output the rest of $S_i$ we just simply increment $j$ until $cost_0(i,j)>\beta$. Note that $X=\cup_i S_i$, and thus to find $X$ we then repeat this procedure for all $i$.
 
Note that in each step of the algorithm we compute a $cost_0$ value, and thus the total time is bounded by the time is takes to compute all the queried $cost_0$ values. For all $n$ values of $i$ we perform a binary search, thus requiring $O(n\log n)$ $cost_0$ queries for all binary searches. For a given $i$, after the binary searching, we then perform $|S_i|$ $cost_0$ queries to determine the rest of the set $S_i$, and thus over all $i$ we perform $|X|=\sum_i |S_i|$ queries. By \lemref{precompConv} each $cost_0$ query takes $O(\log^2 n)$ time, with $O(n\log n)$ preprocessing, and so the total time is thus $O((n\log n+|X|)\log^2 n)$.
\end{proof}
We remark that it should be possible to improve the running time of the above algorithm to $O((n+|X|)\log^2 n)$ using the same approach as in the proof of \thmref{minkfinal}. However, ultimately this will not change the asymptotic running time of our overall algorithm.

\begin{algorithm2e}
    \SetKwInOut{Input}{Input}
    \SetKwInOut{Output}{Output} 
    \Input{An instance $P,k$ of \probref{cxmineps}.}
    \Output{The value $\eps$ of the optimal solution.}
    \DontPrintSemicolon
    Perform the precomputation step from \lemref{precompConv}.\;\label{line:precomp}
    Sample a set $S$ of $4n$ values from $\mathcal{E}$.\;\label{line:sample}
    Sort $S$ and binary search using $Decider$.
    Let $[\alpha,\beta]$ be the resulting interval found where $Decider(\alpha)=False$ and $Decider(\beta)=True$.\;\label{line:interval}
    Let $X=Extract(\alpha,\beta)$.\;\label{line:extracted}
    Sort $X$ and binary search using $Decider$.\;\label{line:finalsearch}
    Return the smallest value $\eps\in X$ such that $Decider$ was $True$.\;
 \caption{Algorithm for solving \probref{cxmineps}.}
\algolab{randomalg}
\end{algorithm2e}

Our algorithm for solving \probref{cxmineps} is shown in \algoref{randomalg}. The correctness of this algorithm is straightforward. 
By the discussion above the optimal value $\eps$ is in $\mathcal{E}$, and the correctness of $Decider$ follows from the previous section. Thus when we binary search over $S$ using $Decider$, we know that $\eps\in [\alpha,\beta]$. Thus, by \lemref{extract}, we know that $X=Extract(\alpha,\beta)$ contains $\eps$. Thus our final binary search over $X$ using $Decider$ is guaranteed to find $\eps$.

The more challenging question is what is the running time of \algoref{randomalg}, for which we have the following helper lemma.

\begin{lemma}\lemlab{whp}
 Let $X=Extract(\alpha,\beta)$ be the set computed on line \ref{line:extracted} in \algoref{randomalg}. Then for any $c\geq 1$, we have that $Pr[|X|> c n\ln n]<1/n^c$.
\end{lemma}
\begin{proof}
Let $\eps$ be the optimal value to the given instance of \probref{cxmineps}.
We first argue that with high probability there are at most $(c/2) n\ln n$ values from $\mathcal{E}$ that are contained in $[\alpha,\beta]$ (i.e.\ in the set $X$) that are also larger than $\eps$. 
Let $Z$ be the $(c/2) n\ln n$ values in $\mathcal{E}$ closest to $\eps$ but also greater than $\eps$. 
(Note that if there are less than $(c/2) n\ln n$ values greater than $\eps$, then the claim trivially holds.) Observe that if our random sample $S$ on line \ref{line:sample} contains even a single value from $Z$ then the claim holds as this value then upper bounds $\beta$, and so there are at most $|Z|=(c/2) n\ln n$ values from $\mathcal{E}$ in $(\eps,\beta]$. 
The probability that the $4n$ sized random sample of values from $\mathcal{E}$ does not contain any element from $Z$ is at most
\[
 (1-|Z|/|\mathcal{E}|)^{4n} \leq (1-((c/2) n\ln n) / n^2)^{4n}
 = (1-(c\ln n)/2n)^{4n} \leq e^{-2c\ln n} = 1/n^{2c}< 1/2n^c,
\]
where we used the standard inequality $1+x\leq e^x$ for any value $x$. 
Note that a symmetric argument yields the same probability bound for the event that there are more than $(c/2)n\log n$ values from $\mathcal{E}$ contained in $[\alpha,\beta]$ that are smaller than $\eps$. Thus by the union bound, the probability that $|X|$ has more than $cn\ln n$ values is less than $1/n^c$.
\end{proof}

\begin{theorem}\thmlab{minepsfinal}
 \algoref{randomalg} solves \probref{cxmineps} in $O(cn\log^3 n)$ time with probability $\geq 1-1/n^c$, for any $c\geq 1$.
\end{theorem}
\begin{proof}
The straightforward correctness of the algorithm has already been discussed above. As for the running time, the precomputation on line \ref{line:precomp} takes $O(n\log n)$ time by \lemref{precompConv}. By \lemref{sample}, it then takes $O(n\log^2 n)$ time to sample the $4n$ values in the set $S$. Sorting $S$ takes $O(n\log n)$ time, and binary searching using $Decider$ takes $O((\log n)\cdot n\log^2 n) = O(n\log^3 n)$ time by \thmref{minkfinal}. By \lemref{extract}, running $Extract(\alpha,\beta)$ on line \ref{line:extracted} to compute $X$ takes $O((n\log n+|X|)\log^2 n)$ time. Finally, sorting and binary searching over $X$ using Decider on line \ref{line:finalsearch} takes $O((n\log^2 n)(\log |X|)+|X|\log |X|)= O((n\log n+|X|)\log^2 n)$, again by \thmref{minkfinal}. 

Thus in total the time is $O((n\log n+|X|)\log^2 n + n\log^3 n)$. By \lemref{whp}, with probability at least $1-1/n^c$ we have $|X|\leq c n\ln n$, and thus with probability at least $1-1/n^c$ the total running time is $O(c n\log^3 n)$. 
\end{proof}

\begin{remark}
 Even in the extremely unlikely event that the algorithm exceeds the $O(n\log^3 n)$ time bound, the worst case running time is only $O(n^2\log^2 n)$. 
\end{remark}


\section{The General Case}\seclab{gencase}
In this section, we remove the restriction that $P$ lies in convex position, showing that \probref{mink} and \probref{mineps} can be solved efficiently by converting them into a corresponding graph problem. 

For any pair of points $a,b\in \Re^2$, define $h_l(a,b)$ to be the closed halfspace bounded by the line going through points $a$ and $b$, picking the halfspace that is to the left of the directed edge $(a,b)$. 
%
Throughout we use $P_{a,b} = P \cap h_l(a,b)$ to denote the subset of $P$ falling in $h_l(a,b)$.

We construct a weighted and fully connected directed graph $G_P=(V,E)$ where $V=P$. 
For an ordered pair of points $(a,b)$ in $P$, the weight of its corresponding directed edge is defined as $w(a,b) = cost(\{a,b\},P_{a,b})$, 
i.e.\ the distance of the furthest point in $P_{a,b}$ from the segment $\overline{ab}$.
(Relating to the previous section, when $P$ is in convex position $w(a,b)=cost_0(a,b)$.)
For a cycle of vertices $C=\{p_1,\ldots,p_k\}$, let $w(C)$ denote the maximum of the weights of the directed edges around the cycle. Throughout, we only consider non-trivial cycles, that is cycles must have at least two vertices. 

The following lemma shows how to compute edge weights. 
We remark that the first half of its proof is nearly identical to that for \lemref{precompConv}, however, the second half differs.

\begin{lemma}\lemlab{precompgeneral}
 Let $P$ be a set of $n$ points in $\Re^2$. Then one can compute $w(a,b)$ for all pairs $a,b\in P$ simultaneously in $O(n^2\log n)$ time.
\end{lemma}
\begin{proof}
 Let $\ell$ denote the line through $a$ and $b$, which we view as being oriented in the direction from $a$ towards $b$.
 Also, let $r_a$ and $r_b$ denote the rays originating at $a$ and $b$ respectively, pointing in the direction orthogonal to $\ell$ and on the side of $\ell$ containing $P_{a,b}$.
 
 Observe that the projection of any point $x\in P_{a,b}$ onto $\ell$ either lies on the portion of $\ell$ before $a$, on the line segment $\overline{ab}$, or on the portion of $\ell$ after $b$. 
 Thus we have a natural partition of $P_{a,b}$ into three sets, the subset in the right angle cone $C_a$ bounded by $\ell$ and $r_a$, those in the slab $Slab(a,b)$ bounded by $\ell$, $r_a$, and $r_b$, and those in the right angle cone $C_b$ bounded by $\ell$ and $r_b$. Observe that for any point $x$ in $C_a$ or $C_b$, its closest point on $\overline{ab}$ is $a$ or $b$, respectively, and moreover $||x-\ell||\leq ||x-\overline{ab}||$. Thus we have that, 
  \begin{align*}
  w(a,b) 
 &= \max\{\max_{x\in C_a\cap P_{a,b}} ||x-a||, \max_{x\in C_b\cap P_{a,b}} ||x-b||, \max_{x\in Slab(a,b)\cap P_{a,b}} ||x-\overline{ab}||\}\\
 &= \max\{\max_{x\in C_a\cap P_{a,b}} ||x-a||, \max_{x\in C_b\cap P_{a,b}} ||x-b||, \max_{x\in P_{a,b}} ||x-\ell||\}.
 \end{align*} 
Therefore, it suffices to describe how to compute each of the three terms in the stated time. To compute $\max_{x\in P_{a,b}} ||x-\ell||$ we use the standard fact that for any point set $P$ and line $\ell$, the furthest point in $P$ from $\ell$, on either side of $\ell$, is a vertex of $\CH(P)$. Thus the furthest point in $P_{a,b}$ from $\ell$ is a point of $\CH(P)$. 
So precompute $\CH(P)$, using any standard $O(n\log n)$ time algorithm, after which we can assume the vertices of $\CH(P)$ are stored in an array sorted in clockwise order. Observe that the subset of the vertices of $\CH(P)$ which are in $P_{a,b}$ is a subarray (or technically two subarrays if it wraps around). So we can determine the ends of this subarray by binary searching. The distances of the points in this subarray to $\ell$ is a concave function, and so we can binary search to find $\max_{x\in P_{a,b}} ||x-\ell||$. These two binary searches take $O(\log n)$ time per pair $a,b$, and thus $O(n^2 \log n)$ time in total.

To compute the $\max_{x\in C_a\cap P_{a,b}} ||x-a||$ values, we do the following (the $b$ values are computed identically). Consider a right angle cone whose origin is at $a$. We conceptually rotate this cone around $a$ while maintaining the furthest point of $P$ from $a$ in this cone. The furthest point only changes when a point enters or leaves the cone, and these events can thus easily be obtained by simply angularly sorting the points in $P$ around $a$. (Note each point corresponds to two events, an entering one, and a leaving one at the entering angle minus $\pi/2$.) To efficiently update the furthest point, we maintain a binary max heap on the distances of the points in the current cone to $a$. Building the initial max heap and sorting takes $O(n\log n)$ time. Thus all possible right angle cone values at $a$ can be computed in $O(n\log n)$ time, as there are a linear number of events and each event takes $O(\log n)$ time. Moreover, if we store these canonical right angle cone values in sorted angular order, then given a query right angle cone determined by a pair $a,b\in P$ (with cone origin $a$), the nearest canonical cone can be determined by binary searching. Thus in total computing all $\max_{x\in C_a\cap P_{a,b}} ||x-a||$ values for all pairs $a$ and $b$ takes $O(n^2 \log n)$ time. Namely, the precomputation of the canonical cones at each point takes $O(n\log n)$ time per point and thus $O(n^2\log n)$ time for all points. Then for the $O(n^2)$ pairs $a,b$ it takes $O(\log n)$ time to search for its canonical cone.
\end{proof}

For a set of points $Q$, let $\CH_L(Q)$ denote the clockwise list of vertices on the boundary of $\CH(Q)$. Observe that any subset $Q\subseteq P$ corresponds to the cycle $\CH_L(Q)$ in $G_P$. 
Moreover, any cycle $C$ corresponds to the convex hull $\CH(C)$. The following lemma is adapted from \cite{krv-schc-20}, where \probref{mink} was considered but where the $cost$ function was determined by a sum of the distances rather than the maximum distance.

\begin{lemma}\lemlab{bounded}
Consider an instance $P,\eps$ of \probref{mink}.
The following holds:
\begin{enumerate}[1)]
 \item For any cycle $C$ in $G_P$, $w(C)\geq cost(C,P) $,
 \item There exists some optimal solution $Q$ such that $w(\CH_L(Q))=cost(Q,P)$.
\end{enumerate}
\end{lemma}
\begin{proof}
Recall that $cost(C,P)=\max_{p \in P} \distX{p}{\CH(C)}$. 
Similarly decomposing $w(C)$ gives,
\begin{align*}
w(C) = \max_{(a,b)\in C} cost(\{a,b\},P_{a,b}) 
= \max_{p\in P}\max_{\substack{(a,b)\in C\\ \text{s.t. } p\in P_{a,b}}}\!\!\!\! \distX{p}{\overline{ab}}.
\end{align*}
To prove the first part of the lemma, we argue that for any point $p\in P$, its contribution to $w(C)$ is at least as large as its contribution to $cost(C,P)$. 
Assume $p\notin \CH(C)$, since otherwise it does not contribute to $cost(C,P)$.
It suffices to argue there exists an edge $(a,b) \in C$, such that $p \in P_{a,b}$, since $||p-\overline{ab}|| \geq ||p-\CH(C)||$.
So assume otherwise that there is some point $p \in P$ such that $p$ lies strictly to the right of all edges in $C$. Create a line $\ell$ that passes through $p$ and any interior point of any edge $(a,b) \in C$, but does not pass through any other point in $P$. 
The line $\ell$ splits the plane into two halfspaces. 
Observe that since $C$ is a cycle, there must be some edge $(c,d)$ of $C$ which also crosses $\ell$, where $c$ is in the same halfspace as $b$ and $d$ in the same halfspace as $a$ (i.e.\ they have opposite orientations with respect to $\ell$). Thus if $(c,d)$ crosses $\ell$ on the same side of $p$ along $\ell$ as the edge $(a,b)$ then $p$ would lie to the left of $(c,d)$, as it lies to the right of $(a,b)$.
On the other hand, if the intersection of $(c,d)$ with $\ell$ lied on the opposite side of $p$ along $\ell$ as the intersection point of $(a,b)$ with $\ell$, then $p\in \CH(\{a,b,c,d\})\subseteq \CH(C)$. Thus either way we have a contradiction.

To prove the second part of the lemma, let $Q$ be some optimal solution. For any $p\in P$, if $p\in \CH(Q)$ then it lies to the right of all edges in $\CH_L(Q)$, and so it does not affect $w(\CH_L(Q))$ or $cost(Q,P)$. So consider a point $p \notin \CH(Q)$. Let $\overline{ab}$ be the closest edge of $\CH(Q)$ (where $b$ follows $a$ in clockwise order). Note that $||p-\CH(Q)|| = ||p-\overline{ab}||$ and $p\in P_{a,b}$, so if $p$ lies to right of all other edges in $\CH_L(Q)$, then its contribution to $w(\CH_L(Q))$ is $||p-\overline{ab}||$. So suppose $p$ lies to the left of some other edge $\overline{cd}$ (note it may be that $b=c$). 
If this happens, then $p$ is in the intersection of the halfspace to the left of the line from $a$ through $b$ and to the left of the line from $c$ through $d$. 
This implies that $b,c\in \CH(\{a,d,p\})$. 
So let $Q'=Q\cup\{p\}\setminus\{b,c\}$. Observe that  $\CH(Q)\subset \CH(Q')$ and $|Q'| \leq |Q|$, and hence $Q'$ is an optimal solution as $Q$ was an optimal solution.
Now we repeat this procedure while there remains such a point $p$ to the left of two edges.  
We repeat this procedure only finitely many times as in each iteration the convex hull becomes larger (i.e.\ $\CH(Q)$ is a strict subset of $\CH(Q')$). 
If $Q$ denotes the hull after the final iteration, then by the above we have $w(\CH_L(Q))=cost(Q,P)$.
\end{proof}

\begin{corollary}\corlab{minkgraph}
Let $P,\eps$ be an instance of \probref{mink}, and let $C^*$ be the cycle with minimum cardinalty among cycles in $G_P$ with $w(C) \leq \eps$. Then $C^*$ is an optimal solution to the given instance of \probref{mink}.
\end{corollary}
\begin{proof}
Using part 1) of \lemref{bounded} we know that $cost(C^*,P) \leq w(C^*) \leq \eps$, so $C^*$ is a solution. Suppose that $C^*$ is not an optimal solution (i.e.\ it is not of minimum cardinality). Then by part 2) of \lemref{bounded}, there exists some optimal solution $Q$ with $|Q| < |C^*|$ such that $w(\CH_L(Q))=cost(Q,P) \leq \eps$. So, there exists a cycle $\CH_L(Q)$ with cost $\leq \eps$ and size less than $|C^*|$, which is a contradiction as $C^*$ had minimal cardinality among such cycles.
\end{proof}

In the following we will reduce our problem to the all pairs shortest path problem on directed unweighted graphs, which we denote as APSP. Let $A(n)$ be the time required to solve APSP. In \cite{z-apspbsrmm-02} it is shown that $A(n)=\tilde{O}(n^{2+\mu})$,%
\footnote{We use the standard convention that $\tilde{O}(f(n))$ denotes $O(f(n)\log^c n)$ for some $c>0$.}
where $\mu$ satisfies the equation $\omega(1,\mu,1)=1+2\mu$, and where $\omega(1,\mu,1)$ is the exponent of multiplication of a matrix of size $n \times n^\mu$ by a matrix of size $n^\mu \times n$. 
\cite{l-farmm-12} shows that $\mu<0.5302$ and thus $A(n)=O(n^{2.5302})$.

\begin{theorem}\thmlab{minkgen}
Any instance $P,\eps$ of \probref{mink} can be solved in time 
\[O(A(n)+n^2 \log n) = O(n^{2.5302}).\]
\end{theorem}
\begin{proof}
By \corref{minkgraph}, in order to solve \probref{mink}, we just need to find a minimum length cycle with weight at most $\eps$ in the graph $G_P$ defined above. By definition a cycle has weight $\leq \eps$ if and only if all of its edge weights are $\leq \eps$. So let $G_P^\eps$ be the unweighted and directed graph obtained from $G_P$ by removing all edges with weight $>\eps$.
Thus the solution to our problem corresponds to the minimum length cycle in this unweighted graph $G_P^\eps$. This can be solved by computing APSP in $G_P^\eps$. Specifically, the solution is determined by the ordered pair $(a,b)$ with the shortest path subject to the directed edge $(b,a)$ existing in $G_P^\eps$ (i.e.\ it is the shortest path that can be completed into a cycle).  

Computing all of the edge weights in $G_P$ can be done in $O(n^2 \log n)$ time by \lemref{precompgeneral}. Converting $G_P$ into $G_P^\eps$ then takes $O(n^2)$ time. 
Given the APSP distances, finding the minimum length cycle takes $O(n^2)$ time by scanning all pairs to check for an edge. APSP on directed unweighted graphs can be solved in $A(n)=O(n^{2.5032})$ time as described above. So, the total time is 
$O(A(n)+n^2 \log n) = O(n^{2.5302})$.
\end{proof}

Let $A_k(n)$ denote the time it takes to solve APSP on directed unweighted graphs where path lengths are bounded by $k$ (i.e.\ the path length is infinite if there is no $k$ length path). \cite{agm-eapsp-97} showed that $A_k(n) = O(n^{\omega}k \log^2 k)$, where $\omega$ is the exponent of (square) matrix multiplication. \cite{cw-mmap-90} showed that $\omega < 2.376$.

\begin{theorem}\thmlab{minepsgen}
Any instance $P,k$ of \probref{mineps} can be solved in time  
\[O(\min\{A(n),A_k(n)\}(\log n) +n^2 \log n) = 
O(\min\{n^{2.5302}, kn^{2.376}\}).\]
\end{theorem}
\begin{proof}
The idea is to binary search using \thmref{minkgen}. Namely, the optimal solution to the instance  $P,k$ of \probref{mineps} has cost $\leq \eps$ if and only the optimal solution to the instance $P,\eps$ of \probref{mink} uses $\leq k$ points. 
Moreover, the weight of any cycle in $G_P$ is determined by the weight of some edge, and thus by the above discussion the optimal solution to the given instance of \probref{mineps} will be the weight of some edge. There are $O(n^2)$ edge weights, which we can enumerate, sort, and binary search over using \thmref{minkgen}. 
Computing all of the edge weights in $G_P$ and sorting them can be done in $O(n^2 \log n)$ time by \lemref{precompgeneral}. Thus by \thmref{minkgen}, the total time is 
$O(A(n)(\log n) +n^2 \log n) = O(n^{2.5302})$
(Note we only compute all edge weights a single time, so each step of the binary search then costs $O(A(n))$ time.)

Alternatively, since we know the value of $k$, we can get a potentially faster time for when $k$ is small, by only considering length at most $k$ paths. Specifically, in each call to our decision procedure (i.e.\ \thmref{minkgen}) instead of computing APSP, compute the APSP restricted to length $k$ paths. Then, by the discussion before the theorem, the running time becomes $O(A_k(n)(\log n) +n^2 \log n) = O((n^{\omega}k \log^2 k)(\log n) +n^2 \log n)) = O(kn^{2.376})$.
\end{proof}

\subsection{Faster Approximations}

While our focus in the paper is on exact algorithms, in this section we show how the results above imply faster approximate solutions for the general case.
First, we show that the results from \secref{convexpos} for points in convex position immediately yield near linear time $2$-approximations for the general case. More precisely, we have the following, where $V(\CH(P))$ denotes the vertices of the convex hull of $P$ 
(and recall $V(\CH(P))\subseteq P$).

\begin{lemma}
 Let $P$ be a point set in the plane. Suppose there exists some subset $Q\subseteq P$ such that $cost(Q,P)\leq \eps$ and $|Q|\leq k$.
 Then there exists a subset $Q'\subseteq V(\CH(P))$ such that $cost(Q',P)\leq \eps$ and $|Q'|\leq 2k$. 
\end{lemma}
\begin{proof}
 Let $Q=\{q_1,\ldots,q_k\}$, where the points are labeled in clockwise order. First, we convert $Q$ into a subset of points on the boundary of $\CH(P)$. Specifically, consider the segment $\overline{q_{i-1}q_{i}}$. Consider the ray with base point $q_{i-1}$, and passing through $q_{i}$, and let $z$ be the point of intersection of this ray with the boundary of $\CH(P)$. Let $Q_z=\{q_1,\ldots,q_{i-1}, z, q_{i+1},\ldots,q_k\}$, and observe that $\CH(Q)\subseteq \CH(Q_z)$ as $q_{i}$ lies on the segment $\overline{q_{i-1}z}$.
 Let $\overline{xy}$ be the edge of $\CH(P)$ which contains $z$, and let $Q_z'=\{q_1,\ldots,q_{i-1}, x,y, q_{i+1},\ldots,q_k\}$. (Note if $z\in V(\CH(P))$ then we set $Q_z'=Q_z$.) Since $z\in \overline{xy}$, we have that $\CH(Q)\subseteq \CH(Q_z)\subseteq \CH(Q_z')$. Thus $cost(Q_z',P)\leq \eps$ and $|Q_z'|\leq k+1$. 
 So if we repeat this procedure for all $i$ then we will end up with a set $Q'$ such that $cost(Q',P)\leq \eps$, $|Q'|\leq 2k$, and $Q'\subseteq V(\CH(P))$. 
\end{proof}

Given an instance $P,\eps$ of \probref{mink}, where the optimal solution $Q$ has size $k$, the above implies there is set $Q'\subseteq V(\CH(P))$ such that $cost(Q',P)\leq \eps$ and $|Q'|\leq 2k$. Such a set can be found using the algorithm of \thmref{minkfinal} for the instance $V(\CH(P)),\eps$ of \probref{cxmink}, as $Q'$ is a candidate solution for this instance. 
Also, recall for $X\subseteq P$, the furthest point in $P$ from $\CH(X)$ is always in $V(\CH(P))$, and so if $cost(X,V(\CH(P)))\leq \eps$ then $cost(X,P)\leq \eps$.

Similarly, given an instance $P,k$ of \probref{mineps}, where the optimal solution $Q$ has cost $\eps$, the above implies there is set $Q'\subseteq V(\CH(P))$ such that $cost(Q',P)\leq \eps$ and $|Q'|\leq 2k$. Such a set can be found using the algorithm of \thmref{minepsfinal} for the instance $V(\CH(P)),2k$ of \probref{cxmineps}, again as $Q'$ is a candidate solution. Thus we have the following.

\begin{theorem}
 Let $P,\eps$ be an instance of \probref{mink}, where the optimal solution $Q$ has size $k$. Then in $O(n\log^2 n)$ time one can compute a set $Q'\subseteq V(\CH(P))$ such that $cost(Q',P)\leq \eps$ and $|Q'|\leq 2k$.
 
 Similarly, let $P,k$ be an instance of \probref{mineps}, where the optimal solution $Q$ has cost $\eps$. Then with probability $\geq 1-1/n^c$, for any constant $c$, in $O(n\log^3 n)$ time one can compute a set $Q'\subseteq V(\CH(P))$ such that $cost(Q',P)\leq \eps$ and $|Q'|\leq 2k$.
\end{theorem}

Finally, we remark that if one allows approximating the best $k$ point solution with $k+1$ points (i.e.\ a $(1+1/k)$-approximation), then our graph algorithms from the previous subsection imply near quadratic time approximations (i.e.\ compared to the theorem above, we are trading near linear running time for approximation quality). The idea is, rather than solving APSP, if we chose an appropriate starting point, we can instead solve for single source shortest paths.
Similar observations have been made before for related problems \cite{abosy-fmcnp-89,ar-hacp-05}.
%
%

For a given instance $P,\eps$ of \probref{mink}, let $Q$ be an optimal solution where $|Q|=k$. Let $p$ be an arbitrary point in $V(\CH(P))$. Let $Q'=Q\cup\{p\}$. Observe that $cost(Q',P)\leq cost(Q,P)\leq \eps$ and $|Q'|\leq k+1$. Thus the optimal solution to this instance of \probref{mink}, but where we require it include $p$, is a valid solution to the instance without this requirement, and uses at most one more point.

Now we sketch how the results from \secref{gencase} directly extend to the case where we want the optimal solution restricted to including $p$.
Specifically, for the analogue of \corref{minkgraph}, let $C^*$ be the minimum cardinality cycle in $G_P$ with weight at most $\eps$ such that the cycle includes $p$. To argue that $C^*$ is an optimal solution to the given instance of \probref{mink} among those which must include the point $p$, we need to extend \lemref{bounded} to require including $p$. Part 1) of the lemma immediately extends. The proof of Part 2) starts with some optimal solution $Q$. 
It then performs a transformation of $Q$ into a set $Q'$ so that points are not to the left of two edges, which one can argue implies $w(\CH_L(Q'))=cost(Q',P)$.
This transformation has the properties that $|Q'|\leq |Q|$ and $\CH(Q)\subseteq \CH(Q')$, and hence $cost(Q',P)\leq cost(Q,P)$, and so since $Q$ was optimal so is $Q'$.   
If instead we perform this same transformation on an optimal solution restricted to containing $p$, call it $X$, then the same argument implies we produce a set $X'$ such that $w(\CH_L(X'))=cost(X',P)$, $|X'|\leq |X|$, and $\CH(X)\subseteq \CH(X')$. Moreover, because $\CH(X)\subseteq \CH(X')$ and $p\in V(\CH(P))$, crucially we have $p\in \CH_L(X')$. Thus the modified \lemref{bounded} and hence \corref{minkgraph}, where $p$ is included, both hold.

To find the optimal solution to \probref{mink} containing $p$, we now use the same approach as in \thmref{minkgen}. The difference now however, is that we only need to compute single source shortest paths in $G_P^\eps$ rather than APSP, since we can use $p$ as our starting point. Let $S(n)$ be the time to compute single source shortest paths. In an unweighted graph using BFS gives $S(n)=O(|E|+|V|) = O(n^2)$. Thus replacing $A(n)$ with $S(n)$ in the running time statement of  \thmref{minkgen} gives $O(S(n)+n^2\log n) = O(n^2 \log n)$.
Similarly, replacing $A(n)$ with $S(n)$  for \thmref{minepsgen} gives  $O(S(n)\log n +n^2\log n) = O(n^2 \log n)$. 
%
Thus we have the following.

\begin{theorem}\thmlab{kapprox}
 Let $P,\eps$ be an instance of \probref{mink}, with optimal solution $Q$. Then in $O(n^2\log n)$ time one can compute a set $Q'\subseteq P$ such that $cost(Q',P)\leq \eps$ and $|Q'|\leq |Q|+1$. 
 
 Similarly, let $P,k$ be an instance of \probref{mineps}, where the optimal solution $Q$ has cost $\eps$. Then in $O(n^2\log n)$ time one can compute a set $Q'\subseteq P$ such that $cost(Q',P)\leq \eps$ and $|Q'|\leq |Q|+1$.
\end{theorem}



\bibliographystyle{plainurl}
\bibliography{refs}%

%
%

\end{document}